\documentclass[journal,onecolumn]{IEEEtran}

\IEEEoverridecommandlockouts

\usepackage{cite}
\usepackage{amsmath,amssymb,amsfonts}
\usepackage{algorithmic}
\usepackage{graphicx}
\usepackage{textcomp}
\usepackage{xcolor}
\usepackage{url}
\usepackage{booktabs}
\usepackage{mathtools}
\usepackage{amsthm}

\theoremstyle{plain}
\newtheorem{theorem}{Theorem}[section]
\newtheorem{proposition}[theorem]{Proposition}
\newtheorem{lemma}[theorem]{Lemma}

\theoremstyle{definition}
\newtheorem{definition}[theorem]{Definition}

\theoremstyle{remark}

\title{A Global Characterization of $f$-Divergences Yielding PSD Mutual-Information Matrices}

\author{\IEEEauthorblockN{Zachary Robertson}
\IEEEauthorblockA{\textit{Computer Science} \\
\textit{Stanford University}\\
Stanford, CA \\
zroberts@stanford.edu}
}

\begin{document}

\maketitle

\begin{abstract}
Given $n$ random variables, when does the matrix of pairwise $f$-mutual informations define a PSD kernel over variables? For convex finite generators $f:(0,\infty)\to\mathbb{R}$ with $f(1)=0$ and finite boundary value $f(0)$, we give a closed characterization up to linear transformation $f\sim f+c(t-1)$, which leaves every $f$-divergence and every $f$-mutual-information matrix unchanged. The matrix $M^{(f)}_{ij}:=I_f(X_i;X_j)$ is PSD for every finite-alphabet family if and only if the normalized representative has a globally convergent expansion $\bar f(t)=\sum_{m\ge2}a_m(t-1)^m$, with $a_m\ge0$, on all of $(0,\infty)$. Sufficiency follows from a replica embedding for monomial generators plus closure under nonnegative mixtures. Necessity first extracts the local Taylor cone at $1$ using biased three-point kernels $H_a$, the Belton--Guillot--Khare--Putinar (BGKP) low-rank Hankel positivity-preserver theorem, and then bootstraps analyticity to the divergence. This is a kernel characterization problem, not a metric one: PSD of the variable-indexed matrix is distinct from Hilbertian properties of divergences between distributions. The result explains why Shannon MI and Jensen--Shannon fail, why $\chi^2$ succeeds, and why non-analytic divergences such as total variation and ReLU are excluded.
\end{abstract}

\section{Introduction}

Given $n$ random variables $X_1,\ldots,X_n$, when does the matrix of pairwise mutual informations
$M_{ij}=I(X_i;X_j)$ define a positive semidefinite (PSD) kernel over variables? This question is fundamental for kernel methods built on dependence measures: factor analysis, independence testing, and feature extraction all benefit when $M\succeq 0$ \cite{steuer2002mutual,reshef2011detecting,ver2015information}. Mutual information also appears in transformer training dynamics: Nichani et al.\ \cite{Nichani2024HowTL} show that attention gradients encode pairwise $\chi^2$-mutual information between tokens. Notably, $\chi^2$-divergence ($f(t)=(t-1)^2$) lies in our PSD-generating cone, so the corresponding finite-alphabet $f$-MI matrices are PSD without a near-independence restriction. Yet Shannon mutual information fails this requirement for $n\ge 4$ \cite{jakobsen2014mi},
and (as we show) this indefiniteness can occur under arbitrarily weak pairwise dependence.

This paper characterizes which $f$-divergences yield PSD mutual-information matrices.
The question is distinct from metric properties of divergences between \emph{distributions}
(e.g., $\sqrt{\text{JS}}$ being a metric \cite{endres2003new,briet2009properties}):
we study the \emph{variable-indexed} matrix $M^{(f)}_{ij}:=I_f(X_i;X_j)$,
not distances between probability measures. An $f$-divergence between distributions $P$ and $Q$ is
$$
D_f(P\|Q)=\sum_x q(x)\,f\!\left(\frac{p(x)}{q(x)}\right),
$$
with \(f\) convex and \(f(1)=0\) \cite{csiszar1967information,csiszar2004information}. We define
$$
I_f(X;Y)=D_f\big(P_{XY}\,\|\,P_X\otimes P_Y\big),
$$
and set \(M^{(f)}_{ij}:=I_f(X_i;X_j)\) whenever these values are finite (in particular, the diagonal is finite under mild conditions such as \(f(0)<\infty\)). We ask for PSD \emph{globally}: for every number of variables and every finite-alphabet family. This is the natural kernel requirement for variable-indexed matrices, because the number of measured variables is determined by the dataset rather than by the divergence.

There is one unavoidable non-identifiability. Composition with a linear transformation does not change any divergence:
\[
D_{f+c(t-1)}(P\|Q)=D_f(P\|Q),
\]
since \(\sum_x q(x)(p(x)/q(x)-1)=0\). Hence the right object is the linear transformation class \(f\sim f+c(t-1)\). Our main result says that, after choosing the normalized representative with no linear term at \(t=1\), the PSD-generating cone is exactly the nonnegative cone generated by \((t-1)^m\), \(m\ge2\), and this representation must hold globally on \((0,\infty)\).

This is a kernel characterization problem, not a metric one: PSD of the \emph{variable-indexed} matrix is distinct from Hilbertian properties of divergences between distributions. A tempting local-only statement would miss the diagonal obstruction, since self-information terms evaluate \(f\) at values of order \(1/p_x\), far from \(1\). The theorem below closes this gap by proving that global PSD forces the local nonnegative Taylor expansion at \(1\) to analytically propagate across all of \((0,\infty)\).

\section{Main Results}

\begin{theorem}[PSD-generating $f$: global iff modulo linear transformation]
\label{thm:main}
Let $f:(0,\infty)\to\mathbb{R}$ be convex and finite, with $f(1)=0$, and assume the finite boundary value $f(0):=\lim_{t\downarrow0}f(t)<\infty$. The following are equivalent.
\begin{enumerate}
\item For every $n\in\mathbb{N}$ and every finite-alphabet family of random variables $(X_1,\dots,X_n)$ for which the entries are finite, the matrix
\[
M^{(f)}_{ij}:=I_f(X_i;X_j)
\]
is positive semidefinite.
\item There exist a scalar $c\in\mathbb{R}$ and coefficients $a_m\ge0$ such that the normalized representative
\[
\bar f(t):=f(t)+c(t-1)
\]
satisfies
\[
\bar f(t)=\sum_{m\ge2}a_m(t-1)^m
\]
for every $t\in(0,\infty)$.
\end{enumerate}
Equivalently, modulo the linear transformation $f\sim f+c(t-1)$, the PSD-generating cone is precisely the cone of global nonnegative mixtures of the monomials $(t-1)^m$, $m\ge2$.
\end{theorem}

The linear transformation is not cosmetic: it is exactly the null direction of $f$-divergences. Once a representative is normalized by \(\bar f'(1)=0\), the coefficients \(a_m\) are uniquely determined. Theorem~\ref{thm:main} has three immediate consequences. First, a single negative Taylor coefficient at $1$ rules out global PSD. Second, convex but non-analytic generators at $1$ (e.g. total variation and ReLU) cannot be PSD-generating. Third, generators such as $\chi^2$ and nonnegative entire mixtures of the monomials succeed for every finite alphabet, without a small-dependence restriction.

\subsection{Proof overview}

Sufficiency is direct: each monomial $f_m(t)=(t-1)^m$ is realized as a Gram matrix by the replica embedding, and nonnegative mixtures preserve PSD. Necessity has two layers. The local layer uses the biased latent family below to form three-point kernels $H_a$; replica forcing and the Belton--Guillot--Khare--Putinar (BGKP) low-rank Hankel theorem imply that several slices $H_a$ have nonnegative power-series expansions. We bootstrap analyticity of $f$ at $1$ from analyticity at three of these slices, and coefficient identification yields $f^{(m)}(1)\ge0$ for all $m\ge2$.

The global layer uses Bernoulli slices. For Bernoulli-$p$ variables with covariance $\rho$, the mutual information is an explicit three-term expression $H_p(\rho)$. Global PSD realizes shifted low-rank Hankel tests around every $\rho_0\in(0,p(1-p))$, so shifted BGKP makes $H_p$ analytic throughout this interval. Algebraically solving the formula for $f$ first propagates the local analytic Taylor coefficients to all $t>1$ and then, using the symmetric $p=1/2$ slice, to all $t\in(0,1)$. Finally, Vivanti--Pringsheim rules out a finite positive radius of convergence for the nonnegative Taylor series, so the local series is the global expansion on all of $(0,\infty)$.

\section{Preliminaries}

Given $n$ random variables $X_1,\ldots,X_n$, define the $f$-mutual-information matrix by $M^{(f)}_{ij}:=I_f(X_i;X_j)$ whenever these quantities are finite.

\begin{definition}[PSD-generating]
We say that $f$ is \emph{PSD-generating} if $M^{(f)}(X_1,\dots,X_n)\succeq0$ for every finite-alphabet family $(X_1,\dots,X_n)$ for which all entries are finite.
\end{definition}

\begin{definition}[linear transformation and normalization]
Write $f\sim g$ if $g(t)=f(t)+c(t-1)$ for some $c\in\mathbb R$. Then $I_g(X;Y)=I_f(X;Y)$ for all finite-alphabet $X,Y$. When $f$ is differentiable at $1$, the normalized representative is $\bar f(t)=f(t)-f'(1)(t-1)$, so that $\bar f(1)=\bar f'(1)=0$.
\end{definition}

\begin{definition}[Pairwise weak dependence]\label{def:pairwise-weak}
For a finite collection of discrete random variables $X_1,\dots,X_n$ and $\delta>0$,
define for each pair $(i,j)$ the pairwise joint-to-product ratio
\[
r_{ij}(x_i,x_j)\;:=\;\frac{p_{X_iX_j}(x_i,x_j)}{p_{X_i}(x_i)\,p_{X_j}(x_j)}
\]
for all $(x_i,x_j)$ with $p_{X_i}(x_i)p_{X_j}(x_j)>0$.
We say $(X_1,\dots,X_n)$ is \emph{$\delta$-pairwise-weakly-dependent} if
$r_{ij}(x_i,x_j)\in(1-\delta,1+\delta)$ for all $i\neq j$ and all $(x_i,x_j)$.
\end{definition}

The pairwise weak-dependence notion is used only in the local extraction step of the proof: global PSD implies the same PSD conclusion for every sufficiently weak pairwise family. For Shannon mutual information ($f(t)=t\log t$), it is known that the MI matrix is PSD for $n\le3$ random variables but counterexamples exist for $n=4$ \cite{jakobsen2014mi}. Theorem~\ref{thm:main} identifies the exact obstruction for all convex finite $f$-generators, modulo linear transformation.

\begin{definition}[Local absolute monotonicity at $1$]\label{def:abs-mono}
We say that $f$ is \emph{locally absolutely monotone at $t=1$ modulo linear transformation} if there exists $\varepsilon>0$ and coefficients $a_m\ge0$ such that, after adding a linear transformation term if necessary,
\[
\bar f(t):=f(t)+c(t-1)=\sum_{m=2}^\infty a_m (t-1)^m,\qquad a_m\ge 0,\ \ |t-1|<\varepsilon.
\]
\end{definition}
\noindent

Equivalently, after normalization, $\bar f^{(m)}(1)\ge0$ for all $m\ge2$.

\section{Proof of Theorem~\ref{thm:main}}

We prove the two directions of Theorem~\ref{thm:main}. Throughout, the linear transformation is harmless because the linear term contributes zero to every $f$-divergence.

\subsection{Sufficiency}

Assume that, after a linear transformation normalization,
$$
f(t)=\sum_{m=2}^{\infty}a_m(t-1)^m,\qquad a_m\ge0,
$$
for every $t\in(0,\infty)$. Since $D_f$ is linear in $f$, it suffices to realize each monomial term as a Gram inner product and then take a nonnegative combination. The monomials $(t-1)^m$ are used as termwise generators; convexity is imposed on the resulting $f$, not on each odd monomial separately.

\begin{proposition}[Replica embedding for monomial generators]
\label{prop:replica}
Fix \(m\ge 2\). For any collection of discrete random variables $X_1, \ldots, X_n$, there exist functions $g_i^{(m)}$ such that
$$
I_{f_m}(X_i; X_j) = \langle g_i^{(m)}, g_j^{(m)} \rangle,
\qquad f_m(t)=(t-1)^m.
$$
Hence the matrix $M^{(m)}_{ij} := I_{f_m}(X_i; X_j)$ is a Gram matrix and therefore positive semidefinite.
\end{proposition}

\begin{proof}
Define the centered and scaled indicators
$$
\phi_i^a(x) := \frac{\mathbf{1}\{x=a\} - p_i(a)}{\sqrt{p_i(a)}},
\qquad
g_i^{(m)} := \sum_{a} \frac{\prod_{r=1}^{m} \phi_i^a(X_i^{(r)})}{p_i(a)^{\frac{m}{2}-1}},
$$
where $X_i^{(1)}, \ldots, X_i^{(m)}$ are i.i.d.\ copies of $X_i$. Expanding \(I_{f_m}\) and using independence across replica blocks gives
\(I_{f_m}(X_i;X_j)=\mathbb{E}[g_i^{(m)}g_j^{(m)}]=\langle g_i^{(m)},g_j^{(m)}\rangle\).
\end{proof}

\begin{lemma}[Nonnegative mixtures preserve PSD]
\label{lem:mixture}
If $f_1, f_2$ are PSD-generating and $\alpha_1, \alpha_2 \geq 0$, then $f = \alpha_1 f_1 + \alpha_2 f_2$ is PSD-generating. Moreover, for any locally finite nonnegative mixture over $\{m \geq 2\}$,
$$
I_f(X_i; X_j) = \sum_{m \geq 2} a_m I_{f_m}(X_i; X_j),
$$
and the linear term contributes nothing.
\end{lemma}

Therefore, fix any finite-alphabet family. The set of joint-to-product ratios contributing to the entries is finite and bounded. Since $\sum_{m\ge2}a_m(t-1)^m$ converges for every $t>0$ with $a_m\ge0$, its radius in $z=t-1$ is infinite; hence it is absolutely convergent at each contributing ratio, including $r=0$ by the finite boundary value. Let $f^{(N)}(t)=\sum_{m=2}^N a_m(t-1)^m$. Proposition~\ref{prop:replica} and Lemma~\ref{lem:mixture} give $M^{(f^{(N)})}\succeq0$, and finite summation over atoms gives entrywise convergence $M^{(f^{(N)})}\to M^{(\bar f)}=M^{(f)}$. Since the PSD cone is closed, $M^{(f)}\succeq0$.

\subsection{Necessity}

Assume $f$ is PSD-generating. We first prove a local statement at $t=1$, then propagate it to the whole positive real line.

\textbf{Step 1: local Taylor extraction at $1$.}
We use the biased latent-variable family from the local argument. Fix $a\in(-1,1)$ and $k\ge1$. Let $J\sim\mathrm{Unif}([k])$ and $S\sim\mathrm{Rademacher}(\pm1)$ be independent, and set $U:=S e_J\in\{\pm e_1,\dots,\pm e_k\}\subset\mathbb R^k$. Given loadings $u_i\in\mathbb R^k$, define $Y_i\in\{\pm1\}$ by
\[
\Pr(Y_i=y\mid U)=\tfrac12\bigl(1+y(a+\langle u_i,U\rangle)\bigr),
\qquad |a|+\|u_i\|_\infty\le1.
\]
Then $\Pr(Y_i=y)=\tfrac12(1+ay)$ and, writing $\eta_i:=\langle u_i,U\rangle$,
\[
\rho_{ij}:=\mathbb E[\eta_i\eta_j]=\frac1k\langle u_i,u_j\rangle.
\]
For $i\ne j$, conditional independence given $U$ gives
\[
\frac{\Pr(Y_i=y_i,Y_j=y_j)}{\Pr(Y_i=y_i)\Pr(Y_j=y_j)}
=1+\frac{\rho_{ij}y_i y_j}{(1+ay_i)(1+ay_j)}.
\]
Thus the off-diagonal entries have the three-point form
\begin{align*}
I_f(Y_i;Y_j)&=:H_a(\rho_{ij})\\
&=\frac{(1+a)^2}{4}f\!\left(1+\frac{\rho_{ij}}{(1+a)^2}\right)
+\frac{(1-a)^2}{4}f\!\left(1+\frac{\rho_{ij}}{(1-a)^2}\right)\\
&\quad +\frac{1-a^2}{2}f\!\left(1-\frac{\rho_{ij}}{1-a^2}\right).
\end{align*}
For the diagonal,
\begin{align*}
I_f(Y_i;Y_i)
&=\frac{(1+a)^2}{4}f\!\left(\frac{2}{1+a}\right)
+\frac{(1-a)^2}{4}f\!\left(\frac{2}{1-a}\right)
+\frac{1-a^2}{2}f(0)=:d_a.
\end{align*}

Let $K_a=[H_a(\rho_{ij})]_{i,j}$ and
\[
\Delta_a=\operatorname{diag}\bigl(d_a-H_a(\rho_{11}),\dots,d_a-H_a(\rho_{nn})\bigr).
\]
Form $R$ conditionally independent replicas of each $Y_i$ given the shared latent $U$. Their $f$-MI matrix has the block form
\[
B_R=J_R\otimes K_a+I_R\otimes\Delta_a.
\]
The following elementary implication is the only property of this block form that we use.

\begin{lemma}[Replica forcing]\label{lem:replica-forcing}
Let $K$ be symmetric and let $\Delta$ be diagonal. If $J_R\otimes K+I_R\otimes\Delta\succeq0$ for every $R\ge1$, then $K\succeq0$.
\end{lemma}

\begin{proof}
Diagonalize $J_R$ as $P^\top J_RP=\operatorname{diag}(R,0,\dots,0)$. Then the block matrix is congruent to $\operatorname{diag}(RK+\Delta,\Delta,\dots,\Delta)$. If $v^\top Kv<0$, then $v^\top(RK+\Delta)v<0$ for all sufficiently large $R$, a contradiction.
\end{proof}

Since $f$ is PSD-generating, Lemma~\ref{lem:replica-forcing} gives $K_a\succeq0$ for every admissible finite loading family. The admissible loadings realize all sufficiently small rank-\(\le3\) Hankel tests.

\begin{lemma}[Low-rank Hankel realization]\label{lem:small-gram}
Fix $a\in[0,1)$. There exists $\rho_a>0$ such that every Hankel matrix $A\in P_n((-\rho_a,\rho_a))$ of rank at most $3$ is realized by the latent model with $k=3$: there are $u_1,\dots,u_n\in\mathbb R^3$ with $|a|+\|u_i\|_\infty\le1$ and $\rho_{ij}=A_{ij}$.
\end{lemma}

\begin{proof}
Factor $A_{ij}=\langle v_i,v_j\rangle$ with $v_i\in\mathbb R^3$ and set $u_i=\sqrt3\,v_i$. Then $\rho_{ij}=\frac13\langle u_i,u_j\rangle=A_{ij}$. Choosing $\rho_a<(1-a)^2/3$ ensures
\[
\|u_i\|_\infty\le\|u_i\|_2=\sqrt{3A_{ii}}<1-a,
\]
so the loadings are admissible.
\end{proof}

Therefore $H_a[A]:=[H_a(A_{ij})]_{i,j}\succeq0$ for every rank-\(\le3\) Hankel $A\in P_n((-\rho_a,\rho_a))$. The BGKP theorem (Theorem~\ref{thm:bgkp}) gives
\[
H_a(z)=\sum_{m\ge0}d_m(a)z^m,
\qquad d_m(a)\ge0,
\qquad |z|<\rho_a.
\]
Applying this for $a=0$, $a=1/3$, and $a=1/2$, and using
\[
q_1=\frac{1-1/3}{1+1/3}=\frac12,
\qquad
q_2=\frac{1-1/2}{1+1/2}=\frac13,
\qquad
\frac{\log q_1}{\log q_2}\notin\mathbb Q,
\]
the three-slice bootstrap (Lemma~\ref{lem:analyticity-bootstrap}) implies that $f$ is analytic at $t=1$.

Now Lemma~\ref{lem:coeff-id} identifies the coefficients:
\[
d_m(a) = \frac{T_m(a)}{m!} f^{(m)}(1), \qquad m \geq 0,
\]
where
\[
T_m(a) = \tfrac{1}{4}\!\left[(1+a)^{2-2m} + (1-a)^{2-2m}\right] - \tfrac{1}{2} (a^2-1)^{1-m}.
\]
A direct parity argument shows $T_1(a)=0$ and $T_m(a)>0$ for all $m\ge2$ and all $a\in(0,1)$ (the sum of reciprocal powers exceeds $2$ when the base ratio exceeds $1$). Hence $d_m(a)\ge0$ and $T_m(a)>0$ imply $f^{(m)}(1)\ge0$ for all $m\ge2$.

\textbf{Step 2: Bernoulli slices are analytic away from the boundary.}
Fix $p\in(0,1)$ and define, for
\[
\rho\in\bigl(-\min\{p^2,(1-p)^2\},\,p(1-p)\bigr),
\]
\begin{align*}
H_p(\rho)
&:=p^2 f\!\left(1+\frac{\rho}{p^2}\right)
+2p(1-p)f\!\left(1-\frac{\rho}{p(1-p)}\right)\\
&\quad +(1-p)^2f\!\left(1+\frac{\rho}{(1-p)^2}\right).
\end{align*}
If $Y,Y'$ are Bernoulli-$p$ variables with covariance $\rho$, then $I_f(Y;Y')=H_p(\rho)$.

\begin{lemma}[Bernoulli-slice realization]\label{lem:bernoulli-slice}
For every $p\in(0,1)$ and every $\rho_0\in(0,p(1-p))$, the function $H_p$ is real analytic in a neighborhood of $\rho_0$.
\end{lemma}

\begin{proof}
Choose $c\in(0,1)$ with $\rho_0=c^2p(1-p)$, and let
\[
B=\begin{cases}
c(1-p),&\text{with probability }p,\\
-cp,&\text{with probability }1-p.
\end{cases}
\]
Then $\mathbb EB=0$ and $\mathbb EB^2=\rho_0$. Moreover $p+B$ stays in $[\gamma,1-\gamma]$ for $\gamma=(1-c)\min\{p,1-p\}>0$.

Let $A\in P_n((-\eta,\eta))$ be Hankel of rank at most $3$, with $\eta>0$ small. Factor $A_{ij}=\langle v_i,v_j\rangle$ in $\mathbb R^r$, $r\le3$, and let $W$ be uniform on $\{\pm\sqrt r e_1,\dots,\pm\sqrt r e_r\}$, independent of $B$. Set $Z_i=\langle v_i,W\rangle$. Then $\mathbb EZ_i=0$ and $\mathbb E[Z_iZ_j]=A_{ij}$; for $\eta$ small, $|Z_i|<\gamma/2$.

Define Bernoulli variables conditionally on $(B,W)$ by
\[
\Pr(Y_i=1\mid B,W)=p+B+Z_i.
\]
Then each $Y_i$ has marginal Bernoulli-$p$, and conditional independence gives
\[
\operatorname{Cov}(Y_i,Y_j)=\mathbb E[(B+Z_i)(B+Z_j)]=\rho_0+A_{ij}.
\]
Thus the off-diagonal kernel is $K=[H_p(\rho_0+A_{ij})]_{i,j}$. Replicating conditionally on $(B,W)$ and applying Lemma~\ref{lem:replica-forcing} gives $K\succeq0$ for every such low-rank Hankel $A$. The shifted BGKP corollary (Lemma~\ref{cor:shifted-bgkp}) applied to $F(\rho)=H_p(\rho)$ implies analyticity of $H_p$ near $\rho_0$.
\end{proof}

\textbf{Step 3: algebraic isolation globalizes analyticity of $f$.}
We first prove analyticity on $(1,\infty)$. Fix $t_0>1$. Choose $p>0$ so small that $\rho_0:=p^2(t_0-1)$ lies in $(0,p(1-p))$ and the two auxiliary arguments
\[
1-\frac{p(t_0-1)}{1-p},
\qquad
1+\frac{p^2(t_0-1)}{(1-p)^2}
\]
lie in the local analytic neighborhood of $1$. For $t$ near $t_0$,
\begin{align*}
H_p(p^2(t-1))
&=p^2 f(t)
+2p(1-p)f\!\left(1-\frac{p(t-1)}{1-p}\right)\\
&\quad +(1-p)^2f\!\left(1+\frac{p^2(t-1)}{(1-p)^2}\right).
\end{align*}
The left side is analytic by Lemma~\ref{lem:bernoulli-slice}, and the two other $f$-terms are analytic by the local seed. Solving for $f(t)$ proves analyticity near $t_0$. Since $t_0>1$ was arbitrary, $f$ is analytic on $(1,\infty)$.

Now fix $t_0\in(0,1)$. Taking $p=1/2$ gives
\[
H_{1/2}\!\left(\frac{1-t}{4}\right)=\frac12 f(2-t)+\frac12 f(t).
\]
The argument $(1-t_0)/4$ lies in $(0,1/4)$, so $H_{1/2}$ is analytic there, and $2-t>1$ is already in the analytic region. Hence
\[
f(t)=2H_{1/2}\!\left(\frac{1-t}{4}\right)-f(2-t)
\]
is analytic near $t_0$. Thus $f$ is analytic on all of $(0,\infty)$.

\textbf{Step 4: Vivanti--Pringsheim globalizes the nonnegative series.}
Let $F(z)=\sum_{m\ge2}a_m z^m$ be the local Taylor series of the normalized representative $\bar f(1+z)$. Its coefficients are nonnegative. If its radius of convergence were finite, say $R<\infty$, then the Vivanti--Pringsheim theorem would force a singularity at the positive boundary point $z=R$ \cite[p.~235]{remmert1991complex}. But $\bar f(1+z)$ is real analytic at $z=R$, since $1+R>0$ and we just proved analyticity on $(0,\infty)$; hence it has a local holomorphic extension there, contradicting that singularity. Therefore $R=\infty$. The series agrees with $\bar f(1+z)$ near $z=0$, and hence by analytic continuation along the real interval it agrees for every $z>-1$, i.e.
\[
\bar f(t)=\sum_{m\ge2}a_m(t-1)^m,
\qquad t\in(0,\infty).
\]
This proves the necessity direction and completes the proof of Theorem~\ref{thm:main}.

\section{Practical Implications}

Theorem~\ref{thm:main} turns the search for PSD-generating divergences into an algebraic test on the normalized Taylor expansion, together with the global analyticity forced by PSD.

\subsection{Examples for Common $f$-divergences}

We retain the direct replica counterexample for non-analytic generators, because it is useful diagnostically. Throughout the calculation we fix $a=\tfrac13$, so the joint-to-product ratio for variables $Y_i,Y_j$ admits the three-point decomposition described above.

\textbf{Total variation / ReLU counterexample.}
Consider four coordinates with
\begin{align*}
u_1&=\tfrac{2}{3}(1,0), \quad
u_2=\tfrac{2}{3}\!\left(\tfrac{1}{\sqrt{2}},\,\tfrac{1}{\sqrt{2}}\right), \\
u_3&=\tfrac{2}{3}(0,1), \quad
u_4=\tfrac{2}{3}\!\left(-\tfrac{1}{\sqrt{2}},\,\tfrac{1}{\sqrt{2}}\right).
\end{align*}
For $f(t)=\tfrac12|t-1|$ and $f(t)=(t-1)_+$, the resulting $4\times4$ kernel can be calculated as follows:
\begin{align*}
K_{1/3}(i,j)&=H_{1/3}(\rho_{ij})\\
&=\tfrac49 f\!\left(1+\tfrac{9}{16}\rho_{ij}\right)
+\tfrac19 f\!\left(1+\tfrac94\rho_{ij}\right)
+\tfrac49 f\!\left(1-\tfrac98\rho_{ij}\right).
\end{align*}
For TVD/ReLU, the weighted sum collapses to $H_{1/3}(z)=\tfrac12|z|$. Here the one-hot latent dimension is $k=2$, so
$\rho_{ij}=\tfrac12\langle u_i,u_j\rangle$. The loadings satisfy $\|u_i\|_\infty\le\tfrac23=1-a$, so the family is admissible. Thus
$$
K_{1/3}=\tfrac12|\rho_{ij}|=\tfrac14|\langle u_i,u_j\rangle|=\frac19
\begin{bmatrix}
1 & \tfrac{\sqrt2}{2} & 0 & \tfrac{\sqrt2}{2}\\
\tfrac{\sqrt2}{2} & 1 & \tfrac{\sqrt2}{2} & 0\\
0 & \tfrac{\sqrt2}{2} & 1 & \tfrac{\sqrt2}{2}\\
\tfrac{\sqrt2}{2} & 0 & \tfrac{\sqrt2}{2} & 1
\end{bmatrix}.
$$
Moreover,
$$
\Delta_{ii}=d_{1/3}-H_{1/3}(\rho_{ii})=4/9-1/9=1/3,
\qquad \Delta_{ij}=0\ (i\ne j).
$$
The eigenvalues are
\begin{align*}
\lambda(K_{1/3})&=\{-0.046,\;0.111,\;0.111,\;0.268\},\\
\lambda(\Delta)&=\{1/3,1/3,1/3,1/3\}.
\end{align*}
Replica amplification beyond $R_{\min}=8$ forces an indefinite $f$-MI matrix. This illustrates the non-analytic obstruction in the theorem.

\textbf{Analytic examples.}
For the Kullback--Leibler generator,
$$
f(t)=t\log t=(t-1)+\tfrac12(t-1)^2-\tfrac16(t-1)^3+\tfrac1{12}(t-1)^4-\cdots.
$$
The linear transformation removes the linear term but not the negative cubic coefficient, so KL/Shannon MI is not PSD-generating. Jensen--Shannon similarly has
\begin{align*}
f(t)&=\tfrac12\!\left(t\log t-(t+1)\log\!\left(\tfrac{t+1}{2}\right)\right)\\
&=\tfrac18(t-1)^2-\tfrac1{16}(t-1)^3+\tfrac7{192}(t-1)^4-\cdots,
\end{align*}
and the negative cubic coefficient rules it out. By contrast, $\chi^2$ has the single monomial generator $f(t)=(t-1)^2$ and is PSD-generating. Non-polynomial examples also exist, e.g.
$$
f(t)=\cosh(t-1)-1=\sum_{m\ge1}\frac{(t-1)^{2m}}{(2m)!},
$$
which has a global nonnegative expansion and hence belongs to the PSD-generating cone.

\textbf{Outlook and scope.}
Theorem~\ref{thm:main} identifies the finite-alphabet PSD-generating cone globally, modulo the linear transformation that leaves every $f$-divergence unchanged. Extensions to continuous models will require additional analytic and measure-theoretic hypotheses (e.g.\ bounded likelihood ratios and integrability), but the finite-alphabet obstruction already explains why PSD is exceptionally brittle for information-theoretic dependence measures.

\section{Conclusion}
We gave a global characterization of PSD-generating $f$ for variable-indexed $f$-mutual-information matrices: PSD for every finite-alphabet family holds iff, after linear normalization, $f$ has a globally convergent expansion on $(0,\infty)$ with nonnegative coefficients from order $2$ onward. The proof combines a replica embedding for monomial generators, a replica-forcing reduction to low-rank Hankel positivity tests, the BGKP theorem, and Bernoulli-slice analytic continuation.

\section{Acknowledgement}

I thank Aishwarya Mandyam and Kirill Acharya for helpful feedback on an early version of this manuscript, and the reviewers for their careful engagement. Their comments prompted the construction of a counterexample showing that local absolute monotonicity alone is not a sufficient criterion for global PSD generation. AI was used during manuscript preparation to check algebra, identify possible gaps, and improve exposition. AI assistance also helped in developing and verifying a
counterexample for the failure of a local converse. All arguments, computations, citations, and final text were reviewed and revised by the author, who is responsible for any errors.

\bibliographystyle{IEEEtran}
\bibliography{example_ref}
\newpage

\appendix

\section{Calculations and Verifications for Proof of Theorem~\ref{thm:main}}

\subsection{Replica tensorization for monomial divergences (Proposition~\ref{prop:replica} details)}
\label{app:replicaTensor}

For $m\in\mathbb{N}$, set $f_m(t)=(t-1)^m$ and define the centered and scaled indicators:
$$
\phi_i^a(x):=\frac{\mathbf{1}\{x=a\}-p_i(a)}{\sqrt{p_i(a)}},\quad
g_i^{(m)}:=\sum_{a}\frac{\prod_{r=1}^{m}\phi_i^a(X_i^{(r)})}{p_i(a)^{\frac{m}{2}-1}}.
$$
Independence across the $m$ replica blocks gives
\begin{align*}
\langle g_i^{(m)},g_j^{(m)}\rangle
&=\sum_{a,b}\frac{\big(\mathbb{E}[\phi_i^a(X_i)\phi_j^b(X_j)]\big)^m}{p_i(a)^{\frac{m}{2}-1}p_j(b)^{\frac{m}{2}-1}} \\
&=\sum_{a,b}\frac{C_{ij}(a,b)^m}{[p_i(a)p_j(b)]^{\frac{m}{2}-1}}
=I_{f_m}(X_i;X_j).
\end{align*}

The second equality can be verified as follows:

\begin{align*}
C_{ij}(a,b) &:= \mathbb{E}\!\left[ \phi_i^a(X_i)\,\phi_j^b(X_j) \right] \\
&= \mathbb{E}\!\left[ \frac{(\delta_{X_i,a} - p_i(a))(\delta_{X_j,b} - p_j(b))}{\sqrt{p_i(a)p_j(b)}} \right] \\
&= \frac{p_{X_i X_j}(a,b) - p_i(a)p_j(b)}{\sqrt{p_i(a)p_j(b)}}.
\end{align*}

Using this result we can obtain the third equality:

$$
\sum_{a,b}\frac{C_{ij}(a,b)^m}{[p_i(a)p_j(b)]^{\frac{m}{2}-1}} = \sum_{a,b} \frac{(p_{i j}(a,b) - p_i(a)p_j(b))^m}{[p_i(a)p_j(b)]^{m-1}}$$

$$
= \sum_{a,b} p_i(a)p_j(b) \left(\frac{p_{i j}(a,b)}{p_i(a) p_j(b)} - 1 \right)^m = I_{f_m}(X_i; X_j).
$$

Thus $M^{(m)}=[I_{f_m}(X_i;X_j)]$ is a Gram matrix and hence PSD. Nonnegative mixtures preserve PSD, and the linear term vanishes since
\begin{align*}
&\sum_{a,b}p_i(a)p_j(b) \left(\frac{p_{ij}(a,b)}{p_i(a)p_j(b)}-1 \right) \\
&\quad =\sum_{a,b}\big(p_{ij}(a,b)-p_i(a)p_j(b)\big)=0.
\end{align*}
Extension to locally finite/infinite mixtures follows by truncation and dominated convergence.

\subsection{One-hot biased-coupling family and the three-point mixture}
\label{app:coupling}

Let $J\sim\mathrm{Unif}([k])$ and $S\sim\mathrm{Rademacher}(\pm1)$ be independent,
and set $U := S e_J\in\{\pm e_1,\dots,\pm e_k\}\subset\mathbb{R}^k$.
Fix a common bias $a\in(-1,1)$ and loading vectors $u_i\in\mathbb{R}^k$.
Define
\[
\Pr(Y_i=y \mid U)
= \tfrac12\big(1+y\,(a+\langle u_i,U\rangle)\big),
\qquad |a|+\|u_i\|_\infty\le 1.
\]

Averaging over $U$ gives $\Pr(Y_i=y)=\tfrac12(1+a\,y)$ since $\mathbb{E}[\langle u_i,U\rangle]=0$.
Using conditional independence given $U$ we obtain
\[
\Pr(Y_i=y_i,Y_j=y_j)
=\tfrac14\Big(1+a(y_i+y_j)+(a^2+\rho_{ij})\,y_i y_j\Big),
\]
where
\[
\rho_{ij} := \mathbb{E}[\langle u_i,U\rangle\langle u_j,U\rangle]
= \mathbb{E}[u_{i,J}u_{j,J}]
= \frac{1}{k}\langle u_i,u_j\rangle.
\]

The product of marginals is 
\[
\Pr(Y_i=y_i)\Pr(Y_j=y_j)
=\tfrac14\big(1+a(y_i+y_j)+a^2 y_i y_j\big).
\]

Hence the joint-to-product ratio is
\[
r_{ij}(y_i,y_j)=1+\frac{\rho_{ij}\,y_i y_j}{(1+a y_i)(1+a y_j)}.
\]

Grouping $(y_i,y_j)\in\{\pm1\}^2$ into three classes by the product $y_i y_j$ yields:
\begin{enumerate}
\item $(+1,+1)$ with weight $\tfrac{(1+a)^2}{4}$ and argument $1+\tfrac{\rho_{ij}}{(1+a)^2}$.
\item $(-1,-1)$ with weight $\tfrac{(1-a)^2}{4}$ and argument $1+\tfrac{\rho_{ij}}{(1-a)^2}$.
\item $y_i\ne y_j$ with total weight $\tfrac{1-a^2}{2}$ and argument $1-\tfrac{\rho_{ij}}{1-a^2}$.
\end{enumerate}
Therefore,

\begin{align*}
I_f(Y_i;Y_j) &=:H_a(\rho_{ij}) \\
&=\frac{(1+a)^2}{4}\,f\Big(1+\frac{\rho_{ij}}{(1+a)^2}\Big) \\
&\quad +\frac{(1-a)^2}{4}\,f\Big(1+\frac{\rho_{ij}}{(1-a)^2}\Big) \\
&\quad +\frac{1-a^2}{2}\,f\Big(1-\frac{\rho_{ij}}{1-a^2}\Big).
\end{align*}

When $a=0$, this reduces to $H_0(z)=\tfrac12(f(1+z)+f(1-z))$.

\subsection{Diagonal entry}
\label{app:diag}

For $i=j$, the ratio is supported only on the diagonal events. Specifically,
\[
r_{ii}(y,y) = \frac{1}{\Pr(Y_i=y)}, \qquad r_{ii}(y,\bar y) = 0.
\]
Thus,
\begin{align*}
I_f(Y_i;Y_i)
&= \frac{(1+a)^2}{4}\,f\!\Big(\tfrac{2}{1+a}\Big)
+ \frac{(1-a)^2}{4}\,f\!\Big(\tfrac{2}{1-a}\Big) \\
&\quad + \frac{1-a^2}{2}\,f(0)
=: d_a,
\end{align*}

which depends only on the common bias $a$ and not on the loading vector $u_i\in\mathbb{R}^k$. In particular, for $a=0$ we obtain
\[
d_0 = \tfrac12\big(f(2) + f(0)\big).
\]

\subsection{Replica block form and the PSD forcing step}
\label{app:replicaBlock}

Let
\[
\begin{gathered}
K_a=\bigl[H_a(\rho_{ij})\bigr]_{i,j},\\
\Delta_a=\operatorname{diag}\!\bigl(d_a-H_a(\rho_{11}),\ldots,d_a-H_a(\rho_{nn})\bigr),\\
\rho_{ii}=\mathbb{E}[\eta_i^2].
\end{gathered}
\]

For \(R\) conditionally independent replicas \(\{Y_i^{(r)}\}_{r=1}^R\) (independent draws given the shared latents \((U_1,\ldots,U_k)\)),
\[
B_R=J_R\otimes K_a + I_R\otimes \Delta_a.
\]

Diagonalizing $J_R$ yields:
$$
(P\otimes I_n)^\top B_R (P\otimes I_n)=\mathrm{diag}(R K_a+\Delta_a,\ \Delta_a,\ldots,\ \Delta_a).
$$

Notice $J_R$, the matrix of all ones, represents scalar multiplication which is a rank-one operation. So we can diagonalize $B_R$ so that $P^T J_R P = \mathrm{diag}(R, \ldots, 0)$. The contribution from $K_a$ or the shared component is:

\begin{align*}
&(P\otimes I_n)^\top (J_R\otimes K_a) (P\otimes I_n) \\
&\quad = (P^T J_R P) (I_n K_a I_n) = \mathrm{diag}( R K_a, \ldots, 0).
\end{align*}

The independent component $I_R \otimes \Delta_a$ then contributes $\mathrm{diag}(\Delta_a, \ldots, \Delta_a)$. 

If the $f$-MI matrix is PSD for all replicated families, then $R K_a+\Delta_a\succeq0$ for all $R$. If some $v$ had $v^\top K_a v<0$, then $v^\top(RK_a+\Delta_a)v<0$ for all sufficiently large $R$, a contradiction. Thus \(K_a\succeq0\).

\subsection{Low-rank Hankel positivity-preserver step}
\label{app:SBCR}

The necessity proof uses a low-rank Hankel positivity-preserver theorem rather than a Schoenberg rescaling argument.

\begin{theorem}[Low-rank Hankel positivity preservers]\label{thm:bgkp}
Let \(I=(-\rho,\rho)\). For a function \(F:I\to\mathbb R\), the following are equivalent:
\begin{enumerate}
\item \(F[-]\) preserves positive semidefiniteness entrywise on all matrices in \(P_n(I)\), for all \(n\).
\item \(F[-]\) preserves positive semidefiniteness entrywise on all Hankel matrices in \(P_n(I)\) of rank at most \(3\), for all \(n\).
\item \(F(x)=\sum_{m\ge0}c_m x^m\) on \(I\), with \(c_m\ge0\).
\end{enumerate}
This is a standard consequence of the Belton--Guillot--Khare--Putinar moment-sequence transform theorem; see \cite[Corollary~6.2]{belton2021moment} and \cite[Theorem~2.13]{khare2025entrywise}.
\end{theorem}

\begin{lemma}[Shifted BGKP corollary]\label{cor:shifted-bgkp}
Let \(F\) be defined on an interval containing \((\rho_0-\eta,\rho_0+\eta)\). If \(A\mapsto [F(\rho_0+A_{ij})]_{i,j}\) preserves PSD on all sufficiently small rank-\(\le3\) Hankel matrices \(A\in P_n((-\eta,\eta))\), for all \(n\), then \(F\) is real analytic in a neighborhood of \(\rho_0\).
\end{lemma}

\begin{proof}
Apply Theorem~\ref{thm:bgkp} to \(G(x):=F(\rho_0+x)\) on a sufficiently small symmetric interval around \(0\).
\end{proof}

\subsection{Analyticity bootstrap from three kernel slices}
\label{app:analyticity-bootstrap}

\begin{lemma}[Analyticity bootstrap from three kernel slices]\label{lem:analyticity-bootstrap}
Let \(g:(-\varepsilon,\varepsilon)\to\mathbb R\) be locally Lipschitz with \(g(0)=0\). For \(a\in(-1,1)\), define
\[
\mathcal H_a(z)
:=
\frac{(1+a)^2}{4}\,g\!\left(\frac{z}{(1+a)^2}\right)
+\frac{(1-a)^2}{4}\,g\!\left(\frac{z}{(1-a)^2}\right)
+\frac{1-a^2}{2}\,g\!\left(-\frac{z}{1-a^2}\right).
\]
Assume:
\begin{enumerate}
\item \(\mathcal H_0\) is real analytic on a neighborhood of \(0\).
\item There exist \(a_1,a_2\in(0,1)\) such that \(\mathcal H_{a_1}\) and \(\mathcal H_{a_2}\) are real analytic on a neighborhood of \(0\).
\item If \(q_i=(1-a_i)/(1+a_i)\), then \(\log q_1/\log q_2\notin\mathbb Q\).
\end{enumerate}
Then \(g\) is real analytic on a neighborhood of \(0\). Consequently, if \(g(u)=f(1+u)-f(1)\), then \(f\) is analytic at \(t=1\).
\end{lemma}

\begin{proof}
Write
\[
e(x)=\frac{g(x)+g(-x)}2,\qquad h(x)=\frac{g(x)-g(-x)}2.
\]
Since \(\mathcal H_0(z)=e(z)\), the even part \(e\) is analytic near \(0\). It remains to prove the odd part \(h\) is analytic. Because \(g\) is locally Lipschitz and \(g(0)=0\), \(h(x)=x\,m(x)\) with \(m\) bounded and even near \(0\).

Fix \(a\in(0,1)\) and set \(q=(1-a)/(1+a)\). A direct computation gives
\[
(1+q)^2\mathcal H_a((1-a)^2x)=g(q^2x)+q^2g(x)+2qg(-qx).
\]
Taking odd parts yields
\[
q^2x\bigl(m(x)-2m(qx)+m(q^2x)\bigr)=q^2x\,D_q^2m(x),
\]
where \(D_qm(x):=m(x)-m(qx)\). Hence \(D_q^2m=S_q\) for an analytic \(S_q\) near \(0\). Boundedness of \(m\) implies \(S_q(0)=0\): otherwise the telescoping differences \(d_n=(D_qm)(q^nx)\) would satisfy \(d_n-d_{n+1}=S_q(q^nx)\) and grow linearly after summation.

Write \(S_q(x)=\sum_{k\ge1}c_kx^k\) and define
\[
M_q^{\rm an}(x)=\sum_{k\ge1}\frac{c_k}{(1-q^k)^2}x^k.
\]
Since \(m\) is even, \(S_q\) and \(M_q^{\rm an}\) are even. Then \(D_q^2M_q^{\rm an}=S_q\). Hence \(N_q:=m-M_q^{\rm an}\) is bounded, even, and satisfies \(D_q^2N_q=0\). With \(\Delta_q=D_qN_q\), the identity \(\Delta_q(qx)=\Delta_q(x)\) and bounded telescoping give \(\Delta_q\equiv0\), so \(N_q(qx)=N_q(x)\).

Apply this first to \(q_1\): \(m=M_{q_1}^{\rm an}+N\) with \(N(q_1x)=N(x)\). Using the \(q_2\)-equation gives \(D_{q_2}^2N=A\), where \(A\) is analytic and \(A(0)=0\). Since \(N(q_1x)=N(x)\), also \(A(q_1x)=A(x)\). Continuity and \(q_1^nx\to0\) imply \(A\equiv0\), hence \(D_{q_2}^2N=0\). The same bounded telescoping argument gives \(N(q_2x)=N(x)\).

Because \(\log q_1/\log q_2\notin\mathbb Q\), the subgroup \(\{q_1^mq_2^n:m,n\in\mathbb Z\}\) is dense in \(\mathbb R_{>0}\). Continuity implies \(N\) is constant on a punctured neighborhood; evenness gives the same constant on both sides of \(0\). Thus \(m\) is analytic, so \(h=xm\) is analytic. Since \(e\) is analytic, \(g=e+h\) is analytic.
\end{proof}

\begin{lemma}[Coefficient identification]\label{lem:coeff-id}
Let $H_a$ be the three-point combination
\begin{align*}
H_a(z)&=\frac{(1+a)^2}{4}\,f\Big(1+\frac{z}{(1+a)^2}\Big) \\
&\quad +\frac{(1-a)^2}{4}\,f\Big(1+\frac{z}{(1-a)^2}\Big) \\
&\quad +\frac{1-a^2}{2}\,f\Big(1-\frac{z}{1-a^2}\Big).
\end{align*}
For $|z|$ small enough,
\[
H_a(z)=\sum_{m\ge0} d_m(a)\,z^m,
\]
with $d_0(a)=0$, $d_1(a)=0$, and $d_m(a)=\frac{T_m(a)}{m!}\,f^{(m)}(1)$ for $m\ge2$, where
\[
T_m(a)=\tfrac14\!\big[(1+a)^{2-2m}+(1-a)^{2-2m}\big]\;-\;\tfrac12(a^2-1)^{1-m}.
\]
\end{lemma}

\begin{proof}
Expand $f(1+u)=\sum_{m\ge0}\frac{f^{(m)}(1)}{m!}u^m$ and substitute \(u=z/(1\pm a)^2\) and \(u=-z/(1-a^2)\) into the three-point formula. The coefficient multiplying \(f^{(m)}(1)z^m/m!\) is exactly \(T_m(a)\).
\end{proof}

Combining the BGKP step with Lemma~\ref{lem:coeff-id} yields $d_m(a) = \frac{T_m(a)}{m!} f^{(m)}(1) \ge 0$ for $m \ge 2$.
Since $T_1(a)=0$ and $T_m(a)>0$ for all $m\ge2$ and $a\in(0,1)$ (Appendix~\ref{app:positivity}), we conclude $f^{(m)}(1)\ge0$ for all $m\ge2$, which is the desired necessity condition.

\subsection{Positivity of $T_m(a)$ for $m\ge2$ and the necessity inequalities}
\label{app:positivity}
Let $u=1+a$, $v=1-a$ (so $u>v>0$ and $uv=1-a^2$). For $m\ge2$, write $k=m-1\ge1$:

\begin{align*}
(uv)^{m-1}T_m(a)
&=\frac14\big(u^{1-m}v^{m-1}+v^{1-m}u^{m-1}\big)+\frac12(-1)^m \\
&=\frac14\big(r^k+r^{-k}\big)-\frac12(-1)^k
\end{align*}
where we define $r:=u/v>1.$ If $k$ is odd, RHS $=\tfrac14(r^k+r^{-k})+\tfrac12>0$. If $k$ is even, since $r>1$ and $k\ge2$, $r^k+r^{-k}>2$, hence RHS $>\tfrac12-\tfrac12=0$. Therefore, for all $a\in(0,1)$,
$$
T_1(a)=0,\qquad T_m(a)>0\ \ \forall\, m\ge2.
$$
Because $d_m(a)\ge0$ and $T_m(a)>0$, we obtain
$$
f^{(m)}(1)\ge0\qquad\forall\, m\ge2.
$$

\subsection{TVD/ReLU example calculations}
\label{app:exampleCalc}

Fix \(a=\tfrac13\). The kernel map is
\[
\begin{aligned}
H_{1/3}(z)
&= \frac{4}{9}\,f\!\left(1+\frac{9}{16}z\right)
 + \frac{1}{9}\,f\!\left(1+\frac{9}{4}z\right)\\
&\quad + \frac{4}{9}\,f\!\left(1-\frac{9}{8}z\right),
\qquad z\in\mathbb{R}.
\end{aligned}
\]
We verify that \(H_{1/3}(z)=\tfrac12|z|\) for both
\(f_{\mathrm{TVD}}(t)=\tfrac12|t-1|\) (total variation) and
\(f_{\mathrm{ReLU}}(t)=(t-1)_+\) (ReLU).

\medskip
\noindent\textbf{TVD.}
Since \(f_{\mathrm{TVD}}(1+\beta z)=\tfrac12|\beta z|\),
\[
\begin{aligned}
H_{1/3}(z)
&=\frac12\!\left(\frac{4}{9}\cdot\frac{9}{16}
               +\frac{1}{9}\cdot\frac{9}{4}
               +\frac{4}{9}\cdot\frac{9}{8}\right)|z| \\
&=\frac12\left(\frac14+\frac14+\frac12\right)|z|
=\frac12|z|.
\end{aligned}
\]

\medskip
\noindent\textbf{ReLU.}
Since \(f_{\mathrm{ReLU}}(1+\beta z)=(\beta z)_+\),
\[
\begin{aligned}
H_{1/3}(z)
&=\left(\frac{4}{9}\cdot\frac{9}{16}
       +\frac{1}{9}\cdot\frac{9}{4}\right)(z)_+
 +\frac{4}{9}\left(-\frac{9}{8}z\right)_+ \\
&=\frac12(z)_+ + \frac12(-z)_+
=\frac12|z|.
\end{aligned}
\]

\medskip
\noindent\textbf{Diagonal correction.}
Recall
\[
d_a=\frac{4}{9}f\!\left(\frac{2}{1+a}\right)
   +\frac{1}{9}f\!\left(\frac{2}{1-a}\right)
   +\frac{4}{9}f(0).
\]
For \(a=\tfrac13\), we have \(\frac{2}{1+a}=\tfrac32\) and \(\frac{2}{1-a}=3\).

For TVD, \(f(\tfrac32)=\tfrac14\), \(f(3)=1\), \(f(0)=\tfrac12\), hence
\[
d_{1/3}=\frac{4}{9}\cdot\frac14+\frac{1}{9}\cdot1+\frac{4}{9}\cdot\frac12
=\frac{4}{9}.
\]
For ReLU, \(f(\tfrac32)=\tfrac12\), \(f(3)=2\), \(f(0)=0\), hence
\[
d_{1/3}=\frac{4}{9}\cdot\frac12+\frac{1}{9}\cdot2+\frac{4}{9}\cdot0
=\frac{4}{9}.
\]

Thus \(d_{1/3}=\tfrac49\) in both cases. Using \(H_{1/3}(z)=\tfrac12|z|\), if
\(\rho_{ii}=\tfrac{2}{9}\) then \(H_{1/3}(\rho_{ii})=\tfrac12|\rho_{ii}|=\tfrac19\),
and therefore
\[
\Delta_{ii}=d_{1/3}-H_{1/3}(\rho_{ii})
=\frac49-\frac19=\frac13.
\]

\end{document}